\documentclass[letterpaper, 10 pt, conference]{ieeeconf}  

\let\proof\relax 
\let\endproof\relax 

\usepackage{bbm}
\usepackage{mathptmx}
\usepackage{graphicx} 
\usepackage{amssymb}
\usepackage{url}
\usepackage{subfigure}
\usepackage{amsthm,amsfonts,bm}
\usepackage{amsmath,mathtools}
\usepackage{graphicx} 
\usepackage{caption,cite}
\usepackage{tikz}
\usetikzlibrary{arrows}

\newtheorem{definition}{Definition}

\newtheorem{lemma}{Lemma}

\newtheorem{assumption}{Assumption}

\newtheorem{proposition}{Proposition}

\newtheorem{theorem}{Theorem}

\newtheorem{corollary}{Corollary}

\newenvironment{hproof}{%
  \proof}{\endproof}

\IEEEoverridecommandlockouts   
\begin{document}
\title{Distributed Stabilization of Two Interdependent Markov Jump Linear Systems with Partial Information}

\author{Guanze Peng, Juntao Chen, and Quanyan Zhu
\thanks{The authors are with the Department of Electrical and Computer Engineering, Tandon School of Engineering,
        New York University, Brooklyn, USA.
        {Email: \tt\small \{guanze.peng,jc6412,qz494\}@nyu.edu}}%
}

\maketitle

\begin{abstract}
In this paper, we study the stabilization of two interdependent Markov jump linear systems (MJLSs) with partial information, where the interdependency arises as the transition of the mode of one system depends on the states of the other system. First, we formulate a framework for the two interdependent MJLSs to capture the interactions between various entities in the system, where the modes of the system cannot be observed directly. Instead, a signal which contains information of the modes can be obtained. Then, depending on the scope of the available system state information (global or local), we design centralized and distributed controllers, respectively, that can stochastically stabilize the overall interdependent MJLS. In addition, the sufficient stabilization conditions for the system under both types of information structure are derived. Finally, we provide a numerical example to illustrate the effectiveness of the designed controllers. 
\end{abstract}


\section{Introduction}
Dynamic systems subject to random abrupt changes in their structures and parameters can be modeled by stochastic jump systems. Particularly, when the random jump process is described by a Markovian process with given transition rates, the system is categorized into the class of Markov jump system. Extensive research and investigations have been done on the stability analysis and (optimal) control design of {Markov jump linear systems (MJLSs)} \cite{c1,c2,c3,c4,c5}. Two common features in the adopted system model in these literature are: (i) the state transition rate matrix is time-invariant, i.e., the transition rate matrix is constant; (ii) the Markov parameters of the transition matrix can be accessed.

However, in real cases, the transition rate matrix of a system can be related to the system state. For example, the failure probability of a wind turbine is related to its used time, level of wear, stress and stiffness on the blades \cite{c7}. 
Thus, the general Markov jump system models considered in \cite{c1,c8} are not directly applicable to these real-world applications. Moreover, the modes of the system often cannot be accessed, such as robot navigation problems, machine maintenance, and planning under uncertainty \cite{c9,c10,c11}. In such cases, the modes can only be inferred from the emitted distorted signals. To address these problems, \cite{c12} has modeled the system as a state-dependent MJLS with partial information, in which the transition rate matrix is time-varying due to the evolution of the dynamical system and the controller only has access to the signals providing partial information of the system modes rather than the modes directly. Note that in all above literature, their focused model contains a single Markov jump system. 

\begin{figure}[!t]
\begin{centering}
\includegraphics[width=1\columnwidth]{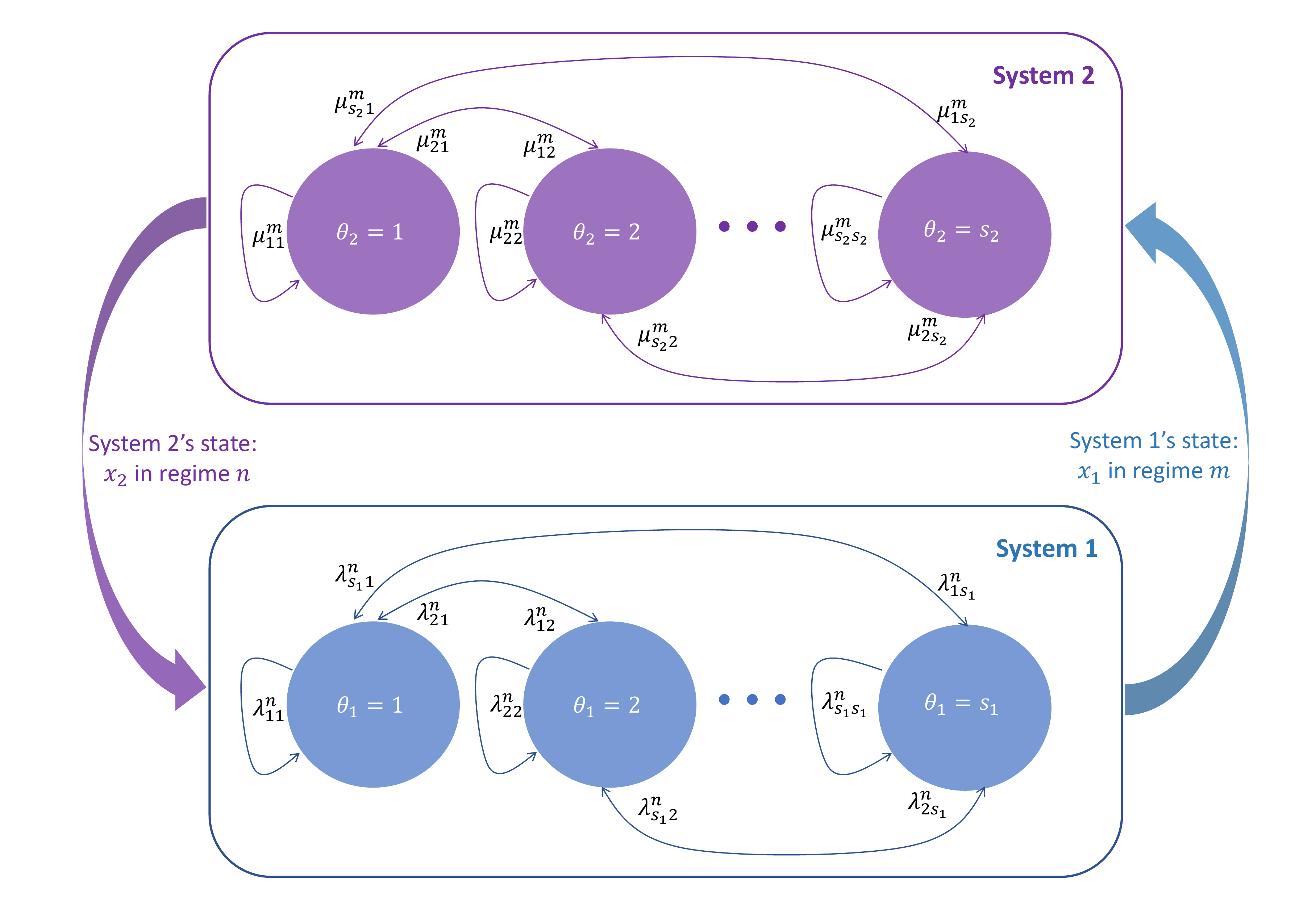}
\end{centering}
\caption{\label{2systemschain}
Two interdependent Markov jump systems model. The operational mode of one system is influenced by the state of the other system.}
\end{figure}

With the emerging of advanced information and communication technologies (ICTs), the real-world systems are becoming more complex. {One main characteristic of these modern control systems is that they are interdependent rather than isolated which forms system-of-systems \cite{c13,c14,c15}. In many examples, we see that the state/condition of one system will have an impact on the operation of other systems. For instance, this particular structure of interdependencies can arise from cascading failures among various entities in the homogeneous and heterogeneous networks \cite{c16,c17,c18}. The Markov jump mode represents the failure mode of an infrastructure. The failure of an infrastructure can be influenced by the state of another. Hence, the cascading failures can be modeled through the interdependencies across mode and state. One specific example lies in the interdependent communication and power systems which can be seen as a class of cyber-physical systems (CPS). The power system operator leverages the communication for real-time control of the grids, while the communication system requires energy support from the power system for functioning purposes. Thus, a point of failure in one system will propagate to the other one due to their interdependencies.} To capture these interdependent features in the networked systems, the traditional single Markov jump system is not sufficient. Therefore, to better understand the interdependencies between different systems and also design controllers for the complex systems, we establish an interdependent MJLS framework in this paper. One illustrative example of two coupled MJLSs is depicted in Fig. \ref{2systemschain}.

Based on the interdependent MJLS model, specifically, we first derive its stability criterion and design stochastic stabilizing controllers by regarding the multiple MJLSs as an integrated system. In addition, in order to preserve the distributed nature of various coupled jump systems, we design the distributed stabilizing controllers for each individual system. 

The main contributions of this paper are summarized as the following.{
\begin{enumerate}
\item We establish an interdependent MJLS model with partial information to capture the interactions and couplings in the complex networks, where Markov mode parameters are state dependent and partially observable.
\item We derive a sufficient stabilization condition, which is in the form of linear matrix inequality (LMI), and design stochastic stable controllers for the integrated MJLS with partial information of the modes.
\item To reduce the complexity of controller design, we design distributed stabilizing controllers for each individual system which ensure the stability of the integrated system-of-systems. 
\end{enumerate}}
 
The rest of the paper is organized as follows. {In Section \ref{prelim}, we describe the interdependent MJLS framework under partial information and the corresponding integrated system.} Section \ref{overallsystem} studies the stabilization problem from an integrated system perspective. Section \ref{individualsystem} investigates the distributed stabilizing controller design. {Simulation studies are given in Section \ref{simulation} to validate the effectiveness of the designed controllers, and finally Section \ref{conclusion} concludes the paper and gives possible directions of future work.}

\section{Interdependent MJLSs and the Integrated System Model}\label{prelim}
In this section, we present the interdependent MJLS framework as shown in Fig. \ref{2systemschain}. Specifically, we consider a model of two coupled MJLSs (System 1 and System 2):
\begin{equation}\label{two_systems}
\begin{aligned}
\dot{x}_k(t) = A_{k,\theta_k(t)}x_k(t)+B_{k,\theta_k(t)}u_k(t)+D_{k,\theta_k(t)}w_k(t),\quad k=1,2,
\end{aligned}
\end{equation}
{
where $x_k(t)\in\mathbb{R}^{N_{k,x}}$, $x_k(t_0)$ is a fixed (known) initial state of the physical plant at starting time $t_0$, $u_k(t)\in\mathbb{R}^{N_{k,u}}$ is the control input, $w_k(t)\in\mathbb{R}^{N_{k,w}}$ is the deterministic disturbance, and all these quantities lie at the physical and control layers of the entire system.} Note that $N_{k,x}$, $N_{k,u}$, and $N_{k,w}$, $k=1,2,$ are all positive integers. Furthermore, the system mode of System $k$, $\theta_k(t)\in\mathbb{R}$, is a Markov jump process with right-continuous sample paths and initial distribution $\pi_{k,0}$. The possible values of $\theta_k(t)$ are assumed to be in the finite set ${\mathcal{S}}_k:=\{1,2,...,|\mathcal{S}_k|\}$. Moreover, $A_{k,\theta_k(t)}, B_{k,\theta_k(t)},$ and $D_{k,\theta_k(t)}$, $k=1,2,$ are system matrices of appropriate dimensions whose entries are continuous functions of time $t$. We assume that the system disturbance $w_k(t)$ satisfies
$
\int_{t_0}^\infty w_k(t)^{\rm T}w_k(t) dt<\infty,\ k=1,2.
$

The MJLSs in \eqref{two_systems} are interdependent in the sense that the transition rate matrix of the system mode of one MJLS is dependent on the state of the other MJLS.  Without loss of generality, we consider the interdependency in a chain structure. Based on the interdependent structure of two MJLSs, we have 
\begin{equation}\label{transition1}
\begin{split}
\mathrm{Pr}[\theta_1(t+\Delta&)=j_1|\theta_1(t)=i_1,x_2(t)\in\mathcal{C}_2^{m_2}] \\
&= \begin{cases}
\begin{array}{cccc}
\lambda_{i_1j_1}^{m_2}\Delta+o(\Delta) &\mathrm{if}\ i_1\neq j_1,\\
1+\lambda_{i_1j_1}^{m_2}\Delta+o(\Delta) &\text{otherwise},
\end{array}\end{cases}
\end{split}
\end{equation}
and
\begin{equation}\label{transition2}
\begin{split}
\mathrm{Pr}[\theta_2(t+\Delta&)=j_2|\theta_2(t)=i_2,x_1(t)\in\mathcal{C}_1^{m_1}] \\
&= \begin{cases}
\begin{array}{cccc}
\mu_{i_2j_2}^{m_1}\Delta+o(\Delta) &\mathrm{if}\ i_2\neq j_2,\\
1+\mu_{i_2j_2}^{m_1}\Delta+o(\Delta) &\text{otherwise},
\end{array}\end{cases}
\end{split}
\end{equation}
where $\mathcal{C}_k^1,\mathcal{C}_k^2,...,\mathcal{C}_k^{M_k}$ , $k=1,2$, are nonempty and disjoint sets, and $\cup_{m_k\in\mathcal{M}_k} \mathcal{C}_k^{m_k}$ expand the space containing all the possible states of $x_k(t)$, where $\mathcal{M}_k:=\{1,2,...,M_k\}$. The transition rates for the Markov jump process, $\theta_1(t)$ and $\theta_2(t)$, are denoted by $\{\lambda_{i_1j_1}^{m_2}\}_{i_1,j_1\in\mathcal{S}_1}$ and $\{\mu_{i_2j_2}^{m_1}\}_{i_2,j_2\in\mathcal{S}_2}$, respectively. {The finite partition of the state space in \eqref{transition1} and \eqref{transition2} is motivated by the interdependent critical infrastructure applications. For example, in the coupled power and communication systems, communication delay will impact the power system operation. Depending on the significance of the delay (minimal, intermediate, enormous), power system operates under different conditions (efficient, delay-tolerant, conservative).}

In the focused scenario, for each MJLS, the system mode $\theta_k(t)$ cannot be directly observed. Instead, a signal $\hat{\theta}_k(t)$ is emitted. At time $t$, given $x_k(t)\in\mathcal{C}_k^{m_k}$ and $\theta_k(t)=i_k$, the observation probabilities are assumed to be the following conditional probabilities: 
\begin{equation*}
\begin{aligned}
&{\rm Pr}\left[\hat{\theta}_k(t)=\hat{i}_k|\theta_k(t)=i_k,x_k(t)\in\mathcal{C}_k^{m_k}\right]\ =\ \alpha^{k,m_k}_{\hat{i}_ki_k},\quad k=1,2,
\end{aligned}
\end{equation*}
{where $\hat{\theta}_k(t)\in\hat{\mathcal{S}}_k$ is the observation of System $k$, and $\hat{\mathcal{S}}_k$ is the set which contains all the possible observations of System $k$.}

{The following assumptions are made to hold throughout this work.
\begin{assumption}\label{assump1}
For each MJLS, we assume that the observation does not influence the transition of the system mode, i.e., for all $\Delta>0$, $\theta_k(t+\Delta)\in\mathcal{S}_k$, $\theta_k(t)\in\mathcal{S}_k$, $x_k(t)\in\mathcal{C}^{m_k}_k$, $m_k\in\mathcal{M}_k$, and $\hat{\theta}_k(t)\in\hat{\mathcal{S}}_k$, $k=1,2,$
    \begin{equation*}
\begin{aligned}
    &\mathrm{Pr}\left[\theta_1(t+\Delta)|\hat{\theta}_1(t),\theta_1(t),x_2(t)\right]=\mathrm{Pr}\left[\theta_1(t+\Delta)|\theta_1(t),x_2(t)\right],\\
    &\mathrm{Pr}\left[\theta_2(t+\Delta)|\hat{\theta}_2(t),\theta_2(t),x_1(t)\right]=\mathrm{Pr}\left[\theta_2(t+\Delta)|\theta_2(t),x_1(t)\right],\\
    &\mathrm{Pr}\left[\{\theta_k(t+\Delta)\}_{k=1,2}|\{\hat{\theta}_k(t),\theta_k(t),x_k(t)\}_{k=1,2}\right]\\
    &\qquad\qquad\qquad=\mathrm{Pr}\left[\{\theta_k(t+\Delta)\}_{k=1,2}|\{\theta_k(t),x_k(t)\}_{k=1,2}\right].
\end{aligned}
    \end{equation*}
\end{assumption}
\begin{assumption}\label{assump2}
For each MJLS, the set of observations is assumed to be the same as the set of system modes, i.e.,
		\begin{equation*}
		{\mathcal{S}}_k=\hat{\mathcal{S}}_k,\quad k=1,2.
		\end{equation*}
\end{assumption}}

The two interdependent MJLSs can be jointly represented as an integrated system, which is given as follows (where we omit the time index for notational clarity):
\begin{equation}\label{sys_inte12}
\begin{aligned}
&\begin{bmatrix}
\dot{x}_1\\
\dot{x}_2
\end{bmatrix} = 
\begin{bmatrix}
A_{1,\theta_1} & 0\\
0 & A_{2,\theta_2}
\end{bmatrix} 
\begin{bmatrix}
x_1\\ x_2
\end{bmatrix}+
\begin{bmatrix}
B_{1,\theta_1}  & 0\\
0 & B_{2,\theta_2} 
\end{bmatrix} 
\begin{bmatrix}
u_1\\ u_2
\end{bmatrix}\\
&\qquad\qquad \qquad\qquad \qquad\qquad   \qquad  +
\begin{bmatrix}
D_{1,\theta_1} & 0\\
0 & D_{2,\theta_2}
\end{bmatrix}
\begin{bmatrix}
w_1 \\ w_2
\end{bmatrix}.
\end{aligned}
\end{equation}
The integrated system \eqref{sys_inte12} can be rewritten more compactly as 
\begin{equation}\label{sys1}
\dot{x}(t) = A_{\theta(t)}x(t)+B_{\theta(t)}u(t)+D_{\theta(t)}w(t),
\end{equation}
where $x:=[x^{\rm T}_1,x^{\rm T}_2]^{\rm T}$, $u:=[u^{\rm T}_1,u^{\rm T}_2]^{\rm T}$,
$w:=[w^{\rm T}_1,w^{\rm T}_2]^{\rm T}$, and $\theta$ denotes the mode of the integrated system determined by $\theta_1$ and $\theta_2$. Besides, 
\begin{equation*}
    A_\theta := \begin{bmatrix}
A_{1,\theta_1} & 0\\
0 & A_{2,\theta_2}
\end{bmatrix},\ \text{and}\ 
B_\theta:=\begin{bmatrix}
B_{1,\theta_1}  & 0\\
0 & B_{2,\theta_2} 
\end{bmatrix}. 
\end{equation*}

{In Section \ref{overallsystem}, we focus on designing the stabilizing controller of the integrated system \eqref{sys_inte12} in a centralized fashion. These results will facilitate the distributed stabilizing controller design of the system \eqref{two_systems} in Section \ref{individualsystem}.}

\section{Stochastic Stability Analysis and Control of the Integrated MJLS}\label{overallsystem}
In this section, we analyze the stability of the integrated system \eqref{sys1} and derive its state-feedback stabilizing control. 
{Note that there exists exact correspondence between \eqref{two_systems} and \eqref{sys1} in terms of system parameters including transition probabilities and observation probabilities. To ease the presentation, we redefine the notations of critical variables succinctly for the integrated system \eqref{sys1}. }


Recall that the integrated system mode $\theta$ is determined by $\theta_1\in \mathcal{S}_1$ and $\theta_2\in \mathcal{S}_2$. Then, we define by $\mathcal{S}:=\{1,2,...,|\mathcal{S}_1||\mathcal{S}_2|\}$ the finite set which contains all the possible system modes $\theta(t)$ of \eqref{sys1}.
Furthermore, let $\mathcal{C}^1,\mathcal{C}^2,...,\mathcal{C}^M$ be nonempty and disjoint sets, and $\cup_{m\in\mathcal{M}}\mathcal{C}^m$ expand the space containing all the possible states of $x(t)$, where $\mathcal{M}:=\{1,2,...,M\}$. As $x(t)$ contains both subsystems' states, its partition into $M$ sets are based on the corresponding partitions of $x_1(t)$ and $x_2(t)$, and thus $M=M_1M_2$.
Similar to \eqref{transition1} and \eqref{transition2}, the transition probabilities of system mode $\theta(t)$ in \eqref{sys1} are given by
\begin{equation}\label{transition}
\begin{split}
\mathrm{Pr}[\theta(t+\Delta&)=j|\theta(t)=i,x(t)\in\mathcal{C}^m] \\
&= \begin{cases}
\begin{array}{cccc}
\gamma_{ij}^m\Delta+o(\Delta) &\mathrm{if}\ i\neq j,\\
1+\gamma_{ij}^m\Delta+o(\Delta) &\text{otherwise},
\end{array}\end{cases}
\end{split}
\end{equation}
where the transition rates for the Markov jump process $\theta(t)$ are denoted by $\{\gamma_{ij}^m\}_{i,j\in\mathcal{S}}$.

Let $\hat{\theta}$ denotes the observation of the integrated system determined by $\hat{\theta}_1$ and $\hat{\theta}_2$. At time $t$, given the state $x(t)\in\mathcal{C}^m$ and the system mode $ \theta(t)=i\in\mathcal{S}$, we denote the probability of observing $\hat{\theta}(t)=\hat{i}\in\hat{\mathcal{S}}$ by $\alpha_{i\hat{i}}^m$, i.e., 
\begin{equation}
    {\rm Pr}\left[\ \hat{\theta}(t)=\hat{i}\ \big|\ \theta(t)=i, x(t)\in\mathcal{C}^m\right]=\alpha_{i\hat{i}}^m,
\end{equation}
where $\hat{\mathcal{S}}$ is the finite set which contains all the observations.

{For each $m\in\mathcal{M}$, we define the following notations which will facilitate the controller design. Due to Assumption \ref{assump2}, $\hat{\mathcal{S}}=\mathcal{S}$ and $[\alpha_{i\hat{i}}^{m}]_{i,\hat{i}\in{\mathcal{S}}}$ is a square matrix. If $[\alpha_{i\hat{i}}^{m}]_{i,\hat{i}\in{\mathcal{S}}}$ is invertible, define
\begin{equation}\label{beta_hat}
    [\beta_{\hat{i}i}^{m}]_{\hat{i},i\in{\mathcal{S}}} = \left([\alpha_{i\hat{i}}^{m}]_{i,\hat{i}\in{\mathcal{S}}}\right)^{-1}.
\end{equation}
Otherwise,
\begin{equation}\label{beta_hat2}
[\beta_{\hat{i}i}^{m}]_{\hat{i},i\in{\mathcal{S}}} = \left([\alpha_{i\hat{i}}^{m}]_{i,\hat{i}\in{\mathcal{S}}}\right)^{\dagger},
\end{equation}
where $(\cdot)^{\dagger}$ stands for the pseudo-inverse of a matrix.  In \eqref{beta_hat} (or \eqref{beta_hat2}), $\beta_{\hat{i}i}^{m}$ is the $(\hat{i},i)$-th entry of the inverse (or pseudo-inverse) of the observation probability matrix formed by $\alpha_{i\hat{i}}^{m}$. }

Next, the definition of stochastic stability of a system is given as follows.

\begin{definition}
The equilibrium point 0 of system \eqref{two_systems} is stochastically stable if for arbitrary $x(t_0)\in\mathbb{R}^{N_x}$, and $\theta(t_0)\in\mathcal{S}$, 
$$\mathbb{E}\left[\int_{t_0}^\infty |x(t)|^2dt\right]<\infty.$$ 
\end{definition}

In this section, we aim to design controllers such that the system \eqref{sys1} is stochastically stable. As the system mode $\theta(t)$ cannot be observed directly, the control inputs can only be designed based on $\hat{\theta}(t)$ and $x(t)$. When $x(t)\in\mathcal{C}^m$ and $\hat{\theta}(t)=\hat{i}$, the control input is designed to be of the following state-feedback linear form:
\begin{equation}\label{controllor}
    u(t)={G_{\hat{\theta}(t)}^m}x(t).
\end{equation}
As shown above, the control gain is dependent on the observation $\hat{\theta}(t)$ and the system state $x(t)$. Using this control input, the closed-loop system can be written as
\begin{equation*}
\dot{x}=A^m_{\theta\hat{\theta}}x+D_{\theta}w,
\end{equation*}
where
$
A^m_{\theta\hat{\theta}}={A}_{\theta}+{B}_{\theta}{G}_{\hat{\theta}}^m.
$

Before deriving the stochastic stability criterion of the targeted system, we present the Dynkin's formula in the following.

\begin{lemma}
Let a random process $(x(t),\theta(t))$ be a Markov process, and its stopping times are denoted by $\tau_0,\tau_1,...$ at step $0,1,...$, respectively. {For Lyapunov function $V(x(t),{\theta}(t))$, the Dynkin's formula admits the following form \cite{wuTAC2013}:
\begin{equation}\label{dynkin}
\begin{aligned}
&\mathbb{E}[V(x(t),{\theta}(t))|x(t_0),{{\theta}}({t_0})] -V(x(t_0),\theta({t_0}))\\
&=\sum_{l=0}^{l^*}\mathbb{E} \bigg[  
\int_{t\wedge\tau_{l}}^{t\wedge\tau_{l+1}}\mathcal{L}V(x(\upsilon),\theta(\upsilon))d\upsilon|x(t\wedge\tau_{l}),\theta({t\wedge\tau_{l}})\bigg],
\end{aligned}
\end{equation}}
where $\tau_{0}=0$, $l=0,1,...,l^*,\ l^*\in[0,\infty]$, $\tau_{l^*}\leq \infty$, and $\mathcal{L}V(x(t),{\theta}(t))$ is the infinitesimal generator given by
\begin{equation*}
\begin{aligned}
&\mathcal{L}V(x(t),\theta(t))\\
&=\lim_{\Delta\rightarrow 0} \frac{1}{\Delta} \Big\{
\mathbb{E}\big[
V(x(t+\Delta),{\theta}({t+\Delta}))|(x(t),\theta(t))
 \big]\\
 &\quad\quad\quad\quad\quad\quad\quad\quad\quad\quad\quad\quad\quad\quad\quad\quad\quad- V(x(t),{\theta(t)})\Big\}.
\end{aligned}
\end{equation*}
\end{lemma}

{Specifically, we choose the Lyapunov function to be of the following quadratic form:
\begin{equation}
    V({x}(t),{\theta}(t))= {x}^{\rm T}(t) P_{\theta(t)}{x}(t),
\end{equation}
where $P_{\theta(t)}$ is a symmetric positive definite matrix.}

{
\begin{lemma}\label{simplify}
Assume that $x(t)\in\mathcal{C}^m$, $\theta(t)=i\in\mathcal{S}$, and $\hat{\theta}(t)=\hat{i}\in{\mathcal{S}}.$ Then, the infinitesimal generator of $V$ is equal to
\begin{equation*}
\begin{aligned}
    &\mathcal{L}V({x}(t),\theta(t))\\ 
    & =  {x}^{\rm T}(t) \left( P_i\bar{A}_{i}^{m} + \bar{A}^{m\rm T}_{i}P_i +  \sum_{j\in{\mathcal{S}}} \gamma^m_{ij}P_j\right){x}(t) + 2{x}^{\rm T}(t) P_i {D}_i{w}(t),
\end{aligned}
\end{equation*}
where
$
    \bar{A}_{i}^{m}=\sum_{\hat{i}\in{\mathcal{S}}}\alpha_{i\hat{i}}^{m}A^m_{i\hat{i}}.
$
\end{lemma}
\begin{proof}
See Appendix \ref{appndix1}.
\end{proof}
}

The following theorem gives a sufficient condition that ensures the stochastic stability of the integrated MJLS under partial information.

\begin{theorem}\label{thm21}
The system \eqref{sys1} can be stochastically stabilized if there exist positive definite matrices $X_i$, $Y^{m}_i$, for all $i\in{\mathcal{S}}$, $m\in\mathcal{M}$, and $\kappa_i>0$, satisfying
\begin{equation}\label{thm1}
\begin{aligned}
X_i {A}_{i}^{{\rm T}} &+ Y^{m{\rm T}}_i{B}_{i}^{\rm T} + {A}_{i} X_i + {B}_iY^{m}_i+ \gamma_{ii}^{m}X_i \\
&+ X_i \left(\sum_{j\in {\mathcal{S}}/\{i\}} \gamma_{ij}^{m}X_j^{-1}\right)X_i + \frac{1}{\kappa_i}{D}_{i}^{\rm T}{D}_i  < 0.
\end{aligned}
\end{equation}
By Schur complement lemma \cite{c20}, \eqref{thm1} is equivalent to
\begin{equation}
\begin{bmatrix}
\mathcal{E}_i^{m}& \Lambda_i^{m}\\
\star & -\mathcal{X}_i
\end{bmatrix}< 0,
\end{equation}
where 
\begin{equation*}
    \begin{aligned}
    \mathcal{E}_i^{m}&:=X_i \bar{A}_{i}^{m{\rm T}} + Y^{m{\rm T}}_i{B}_{i}^{\rm T} + \bar{A}_i^m X_i + {B}_i Y_i^{m} + \gamma_{ii}^{m}X_i + ({1}/{\kappa_i}){D}_i^{{\rm T}}{D}_i,\\
    \Lambda_i^{m} &:= [\sqrt{\gamma_{i1}^{m} }X_i,...,\sqrt{\gamma_{i(i-1)}^{m} }X_i,\sqrt{\gamma_{i(i+1)}^{m} }X_i,...,
    \sqrt{\gamma_{i|{\mathcal{S}}|}^{m} }X_i],\\
    \mathcal{X}_i& := \mathrm{diag}\{X_1,...,X_{i-1},X_{i+1},...,X_{|{\mathcal{S}}|}\}.
    \end{aligned}
\end{equation*}
The control gain is given by 
${G}^{m}_{\hat{i}} =\sum_{i\in{\mathcal{S}}}\beta^{m}_{\hat{i}i}Y^{m}_iX_i^{-1}.$ 
\end{theorem}

\begin{proof} 
Based on the results from \cite{c19}, we obtain, for any $\kappa_i>0$, $i\in{\mathcal{S}}$,
\begin{equation*}
    2{x}^{\rm T}(t) P_i {D}_i{w}(t) \leq 
\frac{1}{\kappa_i} {x}^{\rm T}(t) P_i {D}_i {D}^{{\rm T}}_iP_i {x}(t) + \kappa_i {w}^{\rm T}(t){w}(t).
\end{equation*}
Since $P_i$ is symmetric for all $i\in{\mathcal{S}}$, we have:
\begin{equation}\label{lvsmall2}
    \begin{aligned}
    &\mathcal{L}V({x}(t),\theta(t)) \\
    & = {x}^{\rm T}(t) ( P_i\bar{A}^{m}_i +\bar{A}^{m{\rm T}}_iP_i +  \sum_{j\in{\mathcal{S}}} \gamma^{m}_{ij} P_j){x}(t)+ 2{x}^{\rm T}(t) P_i {D}_i{w}(t)\\
     & \leq {x}^{\rm T}(t) ( P_i\bar{A}^{m}_i + \bar{A}^{m{\rm T}}_iP_i +  \sum_{j\in{\mathcal{S}}} \gamma^{m}_{ij} P_j){x}(t)\\
  &\qquad\qquad\qquad+\frac{1}{\kappa_i} {x}^{\rm T}(t) P_i {D}_i {D}_i^{{\rm T}} P_i {x}(t) + \kappa_i {w}^{\rm T}(t){w}(t)\\
  & = {x}^{\rm T}(t) \big( P_i\bar{A}^{m}_i + \bar{A}^{m{\rm T}}_iP_i +  \sum_{j\in{\mathcal{S}}} \gamma^{m}_{ij} P_j + \frac{1}{\kappa_i}  P_i {D}_i {D}_i^{{\rm T} }P_i\big)x(t) \\
  &\qquad\qquad\qquad\qquad\qquad\qquad\qquad\quad + \kappa_i {w}^{\rm T}(t){w}(t)\\
  & = {x}^{\rm T}(t)\Psi^{m}_i{x}(t)+ \kappa_i {w}^{\rm T}(t){w}(t),
    \end{aligned}
\end{equation}
where $\Psi^{m}_i:= P_i\bar{A}^{m}_i + \bar{A}^{m{\rm T}}_iP_i +  \sum_{j\in{\mathcal{S}}} \gamma^{m}_{ij} P_j + ({1}/{\kappa_i})  P_i {D}_i {{D}_i^{{\rm T}}} P_i$. 
Then, 
\begin{equation}\label{lvsmall}
\begin{aligned}
    \mathcal{L}V({x}(t),\theta(t))- \kappa_i {w}^{\rm T}(t){w}(t) 
    &\leq {x}^{\rm T}(t) \Psi^{m}_i{x}(t)\\
    &\leq r_{\sigma}(\Psi^{m}_i) {x}^{\rm T}(t){x}(t),
\end{aligned}
\end{equation}
{where $r_{\sigma}(\cdot)$ denotes the largest eigenvalue of a matrix.}

Next, we use the LMI technique from \cite{c20}, and choose $X_i=P_i^{-1}$. {By pre- and post- multiplying $P_i\bar{A}_i^{m}+\bar{A}_i^{m{\rm T}}P_i+\sum_{j\in{\mathcal{S}}}\gamma^{m}_{ij}P_j+ ({1}/{\kappa_i})  P_i {D}_i {{D}_i^{{\rm T}}} P_i$ with $X_i$ and setting $Y^{m}_i=\bar{G}_{i}^{m}X_i$, where $\bar{G}_{i}^{m}:=\sum_{\hat{i}\in{\mathcal{S}}}\alpha^{m}_{i\hat{i}}{G}_{\hat{i}}^m$, we observe that if \eqref{thm1} holds, then $\Psi^{m}_i<0$.}

By using Dynkin's formula \eqref{dynkin}, and for any $x(t_0)\in\mathcal{C}^{m_0}$, letting $\{m_0,m_1,...\}$ be the successive clusters of states visited, we obtain
\begin{equation*}
    \begin{aligned}
    &\mathbb{E}[V(x(t),\theta(t))|x(t_0),\theta(t_0)]-V(x(t_0),\theta(t_0))\\
    &= \mathbb{E}\big[  \int_{t_0}^{\tau_{0}}\mathcal{L}V(x(\upsilon),\theta(\upsilon))d\upsilon|x(t_0),\theta(t_0) \big]\\
    &\quad  + \mathbb{E}\big[  \int_{\tau_0}^{\tau_{1}}\mathcal{L}V(x(\upsilon),\theta(\upsilon))d\upsilon|x(\tau_0),\theta(\tau_0)\big]\\
    &\quad  + ...+ \mathbb{E}\big[  \int_{t\wedge\tau_{l^*}}^{t\wedge\tau_{l^*+1}}\mathcal{L}V(x(\upsilon),\theta(\upsilon))d\upsilon|x(t\wedge\tau_{l^*}),\theta(t\wedge\tau_{l^*})\big].
    \end{aligned}
\end{equation*}
Therefore, \eqref{lvsmall2} and \eqref{lvsmall} lead to
\begin{equation*}
    \begin{aligned}
        &\mathbb{E}[V(x(t),\theta(t))|x(t_0),\theta(t_0)]-V(x(t_0),\theta(t_0))\\
        &\leq \max\{r_{\sigma}(\Psi^{m}_i)\}\mathbb{E}\big[  \int_{t_{0}}^{t} x^{\rm T}(\upsilon)x(\upsilon) d\upsilon \big]\\
        &\qquad\qquad\qquad\qquad\qquad\qquad\qquad + \kappa_i \int_{t_{0}}^{t} w^{\rm T}(\upsilon)w(\upsilon) d\upsilon.
    \end{aligned}
\end{equation*}
Hence,
\begin{equation}\label{bound}
    \begin{aligned}
    &-\max\{ r_{\sigma}(\Psi^{m}_i) \}\mathbb{E}\big[  \int_{t_{0}}^{t} x^{\rm T}(\upsilon)x(\upsilon) d\upsilon \big]\\
    &\leq  -\max\{ r_{\sigma}(\Psi^{m}_i) \}\mathbb{E}\big[  \int_{t_{0}}^{t} x^{\rm T}(\upsilon)x(\upsilon) d\upsilon \big]\\
    &\qquad\qquad\qquad\qquad\qquad\qquad+ \mathbb{E}[V(x(t),\theta(t))|x(t_0),\theta(t_0)]\\
     &\leq V(x(t_0),\theta(t_0))+\kappa_i \int_{t_{0}}^{t} w^{\rm T}(\upsilon)w(\upsilon) d\upsilon,
    \end{aligned}
\end{equation}
which leads to
\begin{equation*}
    \begin{aligned}
        &\mathbb{E}\big[  \int_{t_{0}}^{t} x^{\rm T}(\upsilon)x(\upsilon) d\upsilon \big]\\
        &\qquad\qquad\quad\qquad\leq \frac{ V(x(t_0),\theta(t_0))+\kappa_i \int_{t_{0}}^{t} w^{\rm T}(\upsilon)w(\upsilon) d\upsilon}{-\max\{ r_{\sigma}(\Psi^{m}_i) \}}.
    \end{aligned}
\end{equation*}
Letting $t\rightarrow \infty$ implies that $\mathbb{E}\big[  \int_{t_{0}}^{\infty} x^{\rm T}(\upsilon)x(\upsilon) d\upsilon \big]$
is bounded by the right hand side of \eqref{bound}. Therefore, the system is stochastically stable if \eqref{thm1} holds.
\end{proof}

In the case of full information where system's mode $\theta(t)$ is observable, we immediately have the following proposition.

\begin{proposition}
The system can be stochastically stabilized if there exist positive definite matrices $X_i$, $Y_i$, for all $i\in\mathcal{S}$, $m\in\mathcal{M}$, and $\kappa_i>0$, satisfying
\begin{equation}\label{prop}
\begin{aligned}
X_i {A}_i^{{\rm T}} &+ Y_i^{\rm T}{B}_i^{{\rm T}} + {A}_i X_i + {B}_i Y_i + \gamma_{ii}^{m}X_i \\
&+ X_i \left(\sum_{j\in {\mathcal{S}}/\{i\}} \gamma_{ij}^{m}X_j^{-1}\right)X_i + \frac{1}{\kappa_i}{D}_i^{{\rm T}}{D}_i  < 0.
\end{aligned}
\end{equation}
By using Schur complement lemma \cite{c20}, \eqref{prop} is equivalent to
\begin{equation}
\begin{bmatrix}
\mathcal{E}_i^{m}& \Lambda_i^{m}\\
\star & -\mathcal{X}_i
\end{bmatrix}< 0,
\end{equation}
where 
\begin{equation*}
\begin{aligned}
    \mathcal{E}_i^{m}:&=X_i {A}_i^{\rm T} + Y_i^{\rm T}{B}_i^{\rm T} + {A}_i X_i + B_i Y^m_i + \gamma_{ii}^{m}X_i + ({1}/{\kappa_i}){D}_i^{{\rm T}}{D}_i,\\
    \Lambda_i^{m}:&= [\sqrt{\gamma_{i1}^{m} }X_i,...,\sqrt{\gamma_{i(i-1)}^{m} }X_i,\sqrt{\gamma_{i(i+1)}^{m} }X_i,...,
    \sqrt{\gamma_{i|{\mathcal{S}}|}^{m} }X_i],\\
    \mathcal{X}_i:&= \mathrm{diag}\{X_1,...,X_{i-1},X_{i+1},...,X_{|{\mathcal{S}}|}\}.
\end{aligned}
\end{equation*}
Then, when the system mode is $i$, the control gain of the system is given by 
$\tilde{G}_i= Y_iX_i^{-1}$.
\end{proposition}
\begin{hproof}
Note that in the fully observable case, for all $m\in\mathcal{M}$, we obtain
\begin{equation*}
    \alpha^{m}_{i\hat{i}}=\begin{cases}
    1 & \text{when $i=\hat{i}$},\\
    0 & \text{otherwise}.
    \end{cases}
\end{equation*}
Then the proposition is an immediate result from Theorem \ref{thm21}. 
\end{hproof}

\section{Distributed Stabilization of the Interdependent MJLSs}\label{individualsystem}
In this section, We focus on \eqref{two_systems} which includes two interdependent MJLSs.  In Section \ref{overallsystem}, we have studied the stability of the integrated MJLS which requires to know global system's state information. However, due to the distributed structure and different types of the jump systems, obtaining the overall system's information is not always possible/convenient. Thus, to enable the distributed control of the interdependent Markov jump systems, we aim to investigate the criterion that leads to the stochastic stability of each individual system in this section.

{Similar to \eqref{beta_hat} and \eqref{beta_hat2}, for $k=1,2,$ when $[\alpha^{k,m_k}_{i_k\hat{i}_k}]_{i_k,\hat{i}_k\in{\mathcal{S}_k}}$ is invertible, we define
\begin{equation*}
[\beta^{k,m_k}_{\hat{i}_ki_k}]_{\hat{i}_k,i_k\in\mathcal{S}_k}=\left([\alpha^{k,m_k}_{i_k\hat{i}_k}]_{i_k,\hat{i}_k\in\mathcal{S}_k}\right)^{-1}.
\end{equation*}
Otherwise,
\begin{equation*}
[\beta^{k,m_k}_{\hat{i}_ki_k}]_{\hat{i}_k,i_k\in\mathcal{S}_k}=\left([\alpha^{k,m_k}_{i_k\hat{i}_k}]_{i_k,\hat{i}_k\in\mathcal{S}_k}\right)^{\dagger}.
\end{equation*}}

Furthermore, similar to \eqref{controllor}, when $x_1(t)\in\mathcal{C}_1^{m_1}$,  $x_2(t)\in\mathcal{C}_2^{m_2}$, the controllers for the two interdependent MJLSs are given by the following state-feedback form:
\begin{equation}
u_k(t)={G^{m_1,m_2}_{k,\hat{\theta}_k(t)}}x_k(t),\quad k=1,2.
\end{equation}
That is, the control gain of System $k$ is dependent on the observation $\hat{\theta}_k(t)$ and the state pair $(x_1(t),x_2(t))$.

Before we proceeding to the main result of this section, we give the following corollary, which presents how the individual stabilizing control of each system can lead to a stable integrated system.
\begin{corollary}\label{coro1}
The stochastic stability of both MJLSs ensures a stochastically stable integrated system. In addition, for $x_1\in\mathcal{C}_1^{m_1}$, $m_1\in\mathcal{M}_1$, and $x_2\in\mathcal{C}_2^{m_2}$, $m_2\in\mathcal{M}_2$,
the stabilizing control $G^{m_1,m_2}_{2,\hat{i}_2}$ and $G^{m_1,m_2}_{2,\hat{i}_2}$, for all $\hat{i}_1\in{\mathcal{S}}_1$ and $\hat{i}_2\in{\mathcal{S}}_2$, of individual System 1 and System 2 lead to a stable integrated interdependent MJLS \eqref{sys_inte12}.
\end{corollary}
\begin{proof}
First, at time $t$, suppose $(\theta_1(t),\theta_2(t))=(i_1,i_2)$ and $(x_1(t),x_2(t))\in\mathcal{C}_1^{m_1}\times \mathcal{C}_2^{m_2}$. Recall that the individual stabilizing controller of one subsystem is designed by considering all the possible states of the other system. In addition, the two systems satisfy
\begin{equation*}
    \mathbb{E}\int_{t_0}^\infty |x_1(t)|^2dt<\infty,\ 
\mathrm{and}\ \mathbb{E}\int_{t_0}^\infty |x_2(t)|^2dt<\infty.
\end{equation*}
Our goal is to show
$
\mathbb{E}\int_{t_0}^\infty |{x}(t)|^2dt<\infty.
$
First, let the Lyapunov function be of the following form:
\begin{equation*}
    \begin{aligned}
        V({x}(t),\theta(t)) &= {x}^{\rm T}(t) \begin{bmatrix}
        P_{1,i_1}& 0\\
        0 & P_{2,i_2}
        \end{bmatrix} {x}(t)\\
        &={x}_1^{\rm T}(t) P_{1,i_1} {x_1}(t) + {x}_2^{\rm T}(t) P_{2,i_2} {x_2}(t),
    \end{aligned}
\end{equation*}
where $P_{1,i_1}\in\mathbb{R}^{N_{1,x}\times N_{1,x}}$ and $P_{2,i_2}\in\mathbb{R}^{N_{2,x}\times N_{2,x}}$ are real, symmetric and positive definite matrices. Based on Lemma \ref{simplify}, the infinitesimal generator of $V$ is equal to
\begin{equation}\label{interineq}
    \begin{aligned}
    &\mathcal{L}V({x}(t),\theta(t)) \\
    & = {x}_1^{\rm T}(t) \left( P_{1,i_1}\bar{A}^{m_1,m_2}_{1,i_1} + \bar{A}^{m_1,m_2\rm T}_{1,i_1}P_{1,i_1} +  \sum_{j_1\in\mathcal{S}_1} \lambda^{m_2}_{{i_1}{j_1}} P_{1,j_1}\right){x}_1(t)\\
    &\quad+ {x}_2^{\rm T}(t)\left( P_{2,i_2}\bar{A}^{m_1,m_2}_{2,i_2} + \bar{A}^{m_1,m_2\rm T}_{2,i_2}P_{2,i_2} +  \sum_{j_2\in\mathcal{S}_2} \mu^{m_1}_{{i_2}{j_2}} P_{2,j_2}\right){x}_2(t),
    \end{aligned}
\end{equation}
where
$
\bar{A}_{1,i_1}^{m_1,m_2}=\sum_{\hat{i}_1\in{\mathcal{S}}_1}\alpha^{1,m_1}_{i_1\hat{i}_1}\left(A_{1,i_1}+B_{1,i_1}G^{m_1,m_2}_{1,\hat{i}_1}\right),
$
and
$
\bar{A}_{2,i_2}^{m_1,m_2}=\sum_{\hat{i}_2\in{\mathcal{S}}_2}\alpha^{2,m_2}_{i_2\hat{i}_2}\left(A_{2,i_2}+B_{2,i_2}G^{m_1,m_2}_{2,\hat{i}_2}\right).
$

Then, by defining $\Psi_{1,i_1}^{m_1,m_2} : =   P_{1,i_1}\bar{A}^{m_1,m_2}_{1,i_1} + \bar{A}^{m_1,m_2\rm T}_{1,i_1}P_{1,i_1} +  \sum_{j_1\in\mathcal{S}_1} \lambda^{m_2}_{{i_1}{j_1}} P_{1,j_1}$ and $\Psi_{2,i_2}^{m_1,m_2}:=P_{2,i_2}\bar{A}^{m_1,m_2}_{2,i_2} + \bar{A}^{m_1,m_2\rm T}_{2,i_2}P_{2,i_2} +  \sum_{j_2\in\mathcal{S}_2} \mu^{m_1}_{{i_2}{j_2}} P_{2,j_2}$ and using the properties $\Psi_{1,i_1}^{m_1,m_2}< 0$ and $\Psi_{2,i_2}^{m_1,m_2}< 0$ for all $m_1\in\mathcal{M}_1, m_2\in\mathcal{M}_2$,  we further obtain
\begin{equation*}
    \begin{aligned}
        \mathbb{E}&\big[V(x(t),\theta(t))|{x}(t_0),\theta(t_0)\big]-V({x}(t_0),\theta(t_0))\\
&\leq\max_{\substack{m_1\in\mathcal{M}_1\\ m_2\in\mathcal{M}_2}}\{ r_\sigma(\Psi_{1,i_1}^{m_1,m_2}) \}\mathbb{E}\big[  \int_{t_{0}}^{t} x_1^T(\upsilon)x_1(\upsilon) d\upsilon \big]\\
&\quad\quad\ +\max_{\substack{m_1\in\mathcal{M}_1\\ m_2\in\mathcal{M}_2}}\{ r_\sigma(\Psi_{2,i_2}^{m_1,m_2}) \}\mathbb{E}\big[  \int_{t_{0}}^{t} x_2^T(\upsilon)x_2(\upsilon) d\upsilon \big].
    \end{aligned}
\end{equation*}
Reorganizing the terms further yields
\begin{align*}
&-\max_{\substack{m_1\in\mathcal{M}_1\\ m_2\in\mathcal{M}_2}}\{r_{\sigma} (\Psi_{1,i_1}^{m_1,m_2}) \}\mathbb{E}\big[  \int_{t_{0}}^{t} x_1^{\rm T}(\upsilon)x_1(\upsilon) d\upsilon \big]\\
&\qquad\qquad-\max_{\substack{m_1\in\mathcal{M}_1\\ m_2\in\mathcal{M}_2}}\{ r_{\sigma}  (\Psi_{2,i_2}^{m_1,m_2}) \}\mathbb{E}\big[  \int_{t_{0}}^{t} x_2^{\rm T}(\upsilon)x_2(\upsilon) d\upsilon \big]\\
&\leq V({x}(t_0),\theta(t_0))-\mathbb{E}\big[V(x(t),\theta(t))|{x}(t_0),\theta(t_0)\big]\\
&\leq V({x}(t_0),\theta(t_0)).
\end{align*}
In addition, we have
\begin{align*}
&\max_{\substack{m_1\in\mathcal{M}_1\\ m_2\in\mathcal{M}_2}}r_{\sigma} (\Psi_{1,i_1}^{m_1,m_2}) \mathbb{E}\big[  \int_{t_{0}}^{t} x_1^T(\upsilon)x_1(\upsilon) d\upsilon \big]\\
&\qquad\quad\quad+\max_{\substack{m_1\in\mathcal{M}_1\\ m_2\in\mathcal{M}_2}}r_{\sigma} (\Psi_{2,i_2}^{m_1,m_2}) \mathbb{E}\big[  \int_{t_{0}}^{t} x_2^T(\upsilon)x_2(\upsilon) d\upsilon \big]\\
& \leq \max \Big\{\max_{\substack{m_1\in\mathcal{M}_1\\ m_2\in\mathcal{M}_2}}r_{\sigma} (\Psi_{1,i_1}^{m_1,m_2}) ,\max_{\substack{m_1\in\mathcal{M}_1\\ m_2\in\mathcal{M}_2}}r_{\sigma} (\Psi_{2,i_2}^{m_1,m_2})  \Big\}\\
& \qquad\quad\cdot \Bigg(\mathbb{E}\big[  \int_{t_{0}}^{t} x_1^{\rm T}(\upsilon)x_1(\upsilon) d\upsilon \big] +  \mathbb{E}\big[  \int_{t_{0}}^{t} x_2^{\rm T}(\upsilon)x_2(\upsilon) d\upsilon \big] \Bigg)\\
& = \max \Big\{\max_{\substack{m_1\in\mathcal{M}_1\\ m_2\in\mathcal{M}_2}}r_{\sigma} (\Psi_{1,i_1}^{m_1,m_2}) ,\max_{\substack{m_1\in\mathcal{M}_1\\ m_2\in\mathcal{M}_2}}r_{\sigma} (\Psi_{2,i_2}^{m_1,m_2})   \Big\}\\
& \qquad\qquad\qquad\qquad\qquad\ \qquad\qquad\cdot\mathbb{E}\big[  \int_{t_{0}}^{t} x^{\rm T}(\upsilon)x(\upsilon) d\upsilon \big] ,
\end{align*}
which yields
\begin{equation*}
    \begin{aligned}
    &V({x}(t_0),\theta(t_0))\\
    &\geq -\max \Big\{\max_{\substack{m_1\in\mathcal{M}_1\\ m_2\in\mathcal{M}_2}}r_{\sigma} (\Psi_{1,i_1}^{m_1,m_2}) \},\max_{\substack{m_1\in\mathcal{M}_1\\ m_2\in\mathcal{M}_2}}r_{\sigma} (\Psi_{2,i_2}^{m_1,m_2}) \Big\}\\
    & \qquad\qquad\qquad\qquad\qquad\qquad\qquad\quad \cdot\mathbb{E}\big[  \int_{t_{0}}^{t} x^{\rm T}(\upsilon)x(\upsilon) d\upsilon \big] \\
    &:=r_{\rm max }\cdot\mathbb{E}\big[  \int_{t_{0}}^{t} x^{\rm T}(\upsilon)x(\upsilon) d\upsilon \big].
    \end{aligned}
\end{equation*}
Thus, we obtain
\begin{align*}
\mathbb{E}\big[  \int_{t_{0}}^{t} x^T(\upsilon)x(\upsilon) d\upsilon \big]\leq \frac{V({x}(t_0),\theta(t_0))}{r_{\rm max}}.
\end{align*}
This completes the proof.
\end{proof}

The following theorem provides sufficient conditions for the integrated MJLS with stabilizing controllers designed in the distributed fashion. 

\begin{theorem}\label{distributed_controller}
The integrated MJLS can be stochastically stabilized if there exist positive definite matrices $X_{k,i_k}>0$,  $Y^{m_1,m_2}_{k,i_k}>0$, for all $i_k\in\mathcal{S}_k$, $m_k\in\mathcal{M}_k$, and $\kappa_{k,i_k}>0$, $k=1,2$, satisfying
\begin{equation*}\label{sec2thm1}
\begin{aligned}
&X_{1,i_1} {A}^{{\rm T}}_{1,i_1} + Y^{m_1,m_2{\rm T}}_{1,i_1}{B}_{1,i_1}^{{\rm T}} + {A}_{1,i_1} X_{1,i_1} + {B}_{1,i_1}Y^{m_1,m_2}_{1,i_1}+ \lambda_{i_1i_1}^{m_2}X_{1,i_1} \\
&\ + X_{1,i_1} \left(\sum_{j_1\in{\mathcal{S}_1}/\{i_1\}} \lambda_{i_1j_1}^{m_2}(X_{1,j_1})^{-1}\right)X_{1,i_1} + \frac{1}{\kappa_{1,i_1}}D_{1,i_1}^{{\rm T}}{D}_{1,i_1}  < 0,\\
&X_{2,i_2} {A}^{{\rm T}}_{2,i_2} + Y^{m_1,m_2{\rm T}}_{2,i_2}{B}_{2,i_2}^{{\rm T}} + {A}_{2,i_2} X_{2,i_2} + {B}_{2,i_2}Y^{m_1,m_2}_{2,i_2}+ \mu_{i_2i_2}^{m_1}X_{2,i_2} \\
&\ + X_{2,i_2} \left(\sum_{j_2\in{\mathcal{S}_2}/\{i_2\}} \mu_{i_2j_2}^{m_1}(X_{2,j_2})^{-1}\right)X_{2,i_2} + \frac{1}{\kappa_{2,i_2}}D_{2,i_2}^{{\rm T}}{D}_{2,i_2}  < 0,
\end{aligned}
\end{equation*}
which is equivalent to 
\begin{equation*}
\begin{bmatrix}
\mathcal{E}_{1,i_1}^{m_1,m_2} & \Lambda_{1,i_1}^{m_2}\\
\star & -\mathcal{X}_{1,i_1}
\end{bmatrix}< 0,\quad\text{and}\quad  \begin{bmatrix}
\mathcal{E}_{2,i_2}^{m_1,m_2} & \Lambda_{2,i_2}^{m_1}\\
\star & -\mathcal{X}_{2,i_2}
\end{bmatrix}< 0,
\end{equation*}
where 
\begin{equation*}
	\begin{aligned}
	&\mathcal{E}_{1,i_1}^{m_1,m_2}:=X_{1,i_1} {A}_{1,i_1}^{{\rm T}} + Y_{1,i_1}^{m_1,m_2{\rm T}}{B}_{1,i_1}^{{\rm T}} + {A}_{1,i_1} X_{1,i_1} \\
	&\qquad\qquad\quad\quad+ {B}_{1,i_1} Y_{1,i_1}^{m_1,m_2} + \lambda_{i_1i_1}^{m_2}X_{1,i_1}+ ({1}/{\kappa_{1,i_1}}){D}_{1,i_1}^{{\rm T}}{D}_{1,i_1},\\
	&\mathcal{E}_{2,i_2}^{m_1,m_2}:=X_{2,i_2} {A}_{2,i_2}^{{\rm T}} + Y_{2,i_2}^{m_1,m_2{\rm T}}{B}_{2,i_2}^{{\rm T}} + {A}_{2,i_2} X_{2,i_2} \\
	&\qquad\qquad\quad\quad+ {B}_{2.i_2} Y_{2,i_2}^{m_1,m_2} + \mu_{i_2i_2}^{m_1}X_{2,i_2}+ ({1}/{\kappa_{2,i_2}}){D}_{2,i_2}^{\rm T}{D}_{2,i_2},\\
&\Lambda_{1,i_1}^{m_2} := [\sqrt{\lambda_{i_11}^{m_2}}X_{1,i_1},...,\sqrt{\lambda_{i_1(i_1-1)}^{m_2}}X_{1,i_1},\\
	&\qquad\qquad\qquad\qquad\qquad\quad\sqrt{\lambda_{i_1(i_1+1)}^{m_2}}X_{1,i_1},..., \sqrt{\lambda_{i_1|\mathcal{S}_1|}^{m_2}}X_{1,i_1}],\\
	&\Lambda_{2,i_2}^{m_1} := [\sqrt{\mu_{i_21}^{m_1}}X_{2,i_2},...,\sqrt{\mu_{i_2(i_2-1)}^{m_1}}X_{2,i_2}, \\
	&\qquad\qquad\qquad\qquad\qquad\quad\sqrt{\mu_{i_2(i_2+1)}^{m_1}}X_{2,i_2},..., \sqrt{\mu_{i_2|\mathcal{S}_2|}^{m_1}}X_{2,i_2}],\\
	&\mathcal{X}_{1,i_1} := \mathrm{diag}\{ X_{1,1},...,X_{1,i_1-1},X_{1,i_1+1},...,X_{1,|\mathcal{S}_1|}\}, \\ &\mathcal{X}_{2,i_2}:= \mathrm{diag}\{ X_{2,1},...,X_{2,i_2-1},X_{2,i_2+1},...,X_{2,|\mathcal{S}_2|}\}. 
\end{aligned}
\end{equation*}
Moreover, the control gain for System k is 
\begin{equation*}
G_{k,\hat{i}_k}^{m_1,m_2}= \sum_{i_k\in\mathcal{S}_k}\beta^{k,m_k}_{\hat{i}_ki_k} Y_{k,i_k}^{m_1,m_2}(X_{k,i_k})^{-1},\ \ k=1,2,
\end{equation*}
for all $ \hat{i}_k\in{\mathcal{S}}_k$.

\end{theorem}
\begin{proof}
The proof straightforwardly follows from Theorem \ref{thm21} and Corollary \ref{coro1}.
\end{proof}

{\textit{Remark:} By comparing the designed stabilizing controllers in Sections \ref{overallsystem} and \ref{individualsystem}, we can find that the number of controllers is different in these two scenarios. Specifically, it requires $M_1M_2|\mathcal{S}_1||\mathcal{S}_2|$ number of controllers through the centralized design method (Section \ref{overallsystem}), while the distributed one reduces it to $M_1M_2(|\mathcal{S}_1|+|\mathcal{S}_2|)$ (Section \ref{individualsystem}), which simplifies the complexity of control design.}

\section{Numerical Experiments}\label{simulation}
In this section, we present a numerical example to illustrate the obtained analytical results. The parameters of the system are $\theta_1\in\mathcal{S}_1=\{1,2\}$ and $\theta_2\in\mathcal{S}_2=\{1,2,3\}$. The system matrices of the independent MJLSs are given as follows:
\begin{equation*}
    \begin{aligned}
        &A_{1,1}=\begin{bmatrix}
        5 & 2\\
        2 & 4
        \end{bmatrix},\ 
        A_{1,2}=\begin{bmatrix}
        5 & 2\\
        2 & 4
        \end{bmatrix},\\
        &B_{1,1}= \begin{bmatrix}
        1\\
        2
        \end{bmatrix},\ 
        B_{1,2} = \begin{bmatrix}
        2\\
        1
        \end{bmatrix},\ A_{2,1}=\begin{bmatrix}
        3 & 2 & 4\\
        5 & 2 & 6\\
        -9 & 0 & 2
        \end{bmatrix},\\
        &
        A_{2,2}=\begin{bmatrix}
        1 & 2 & 3\\
        2 & 1 & 0\\
        5 & 6 & 3
        \end{bmatrix},\ 
        A_{2,3}=\begin{bmatrix}
        4 & -1 & 8\\
        5 & 8 & 0\\
        -1 & 7 & 5
        \end{bmatrix},\\
        &B_{2,1}=\begin{bmatrix}
        1\\
        2\\
        1
        \end{bmatrix},\ 
        B_{2,2}\begin{bmatrix}
        1\\
        0\\
        1
        \end{bmatrix},\ 
        B_{2,3}=\begin{bmatrix}
        2\\
        1\\
        0
        \end{bmatrix}.
    \end{aligned}
\end{equation*}
In addition, the transition rate matrices are
\begin{equation*}
    \begin{aligned}
        &\lambda^1 = \begin{bmatrix}
        -0.6 & 0.6\\
        -0.4 & 0.4
        \end{bmatrix},
        \lambda^2 = \begin{bmatrix}
        -0.2 & 0.2\\
        -0.8 & 0.8
        \end{bmatrix},
        \lambda^3 = \begin{bmatrix}
        -0.5 & 0.5\\
        -1.2 & 1.2
        \end{bmatrix},\\
        &\mu^1=\begin{bmatrix}
        -0.8 & 0.2 & 0.6\\
        0.2 & -0.9 & 0.7\\
        0.5 & 0.4 & -0.9
        \end{bmatrix},
        \mu^2=\begin{bmatrix}
        -0.4 & 0.2 & 0.2\\
        0.2 & -0.5 & 0.4\\
        0.5 & 0.6 & -1.1
        \end{bmatrix}.
    \end{aligned}
\end{equation*}
Specifically, $\lambda^1$, $\lambda^2$ and $\lambda^3$ are transition rate matrices of System 1 under the conditions of $x_2\in\mathcal{C}_2^1=\{x_2:|x_2|^2<5\}$, $x_2\in\mathcal{C}_2^2=\{x_2:5\leq |x_2|^2\leq 10\}$, and $x_2\in\mathcal{C}_2^3=\{x_2:|x_2|^2>10\}$, respectively. Similarly, $\mu^1$ and $\mu^2$ are transition rate matrices of System 2 under the conditions of $x_1\in\mathcal{C}_1^1=\{x_1:|x_1|^2<10\}$, and, $x_1\in\mathcal{C}_1^2=\{x_1:|x_1|^2\geq 10\}$,  respectively.

Moreover, the observation matrices of System 1 and System 2 are given by $P^{m_1}=[\alpha_{i_1\hat{i}_1}^{1,m_1}]_{i_1,\hat{i}_1\in\mathcal{S}_1}$, $m_1=1,2$. and $Q^{m_2}=[\alpha_{i_2\hat{i}_2}^{2, m_2}]_{i_2,\hat{i}_2\in\mathcal{S}_2}$, $m_2=1,2,3$, respectively, with matrices taking the following forms:
$$
    \begin{aligned}
        &P^1 = \begin{bmatrix}
        0.9 & 0.1\\
        0.1 & 0.9
        \end{bmatrix},\ 
        P^2 = \begin{bmatrix}
        0.7 & 0.3\\
        0.3 & 0.7
        \end{bmatrix},\\
        &Q^1 = \begin{bmatrix}
        0.8 & 0.1 & 0.1\\
        0.1 & 0.8 & 0.1\\
        0.1 & 0.1 & 0.8
        \end{bmatrix},\ 
        Q^2 = \begin{bmatrix}
        0.7 & 0.2 & 0.1\\
        0.2 & 0.7 & 0.1\\
        0.2 & 0.1 & 0.7
        \end{bmatrix},\\
        &Q^3 = \begin{bmatrix}
        0.7 & 0.1 & 0.2\\
        0.1 & 0.7 & 0.2\\
        0.1 & 0.2 & 0.7
        \end{bmatrix}.
    \end{aligned}
$$

The stabilizing controllers are designed by solving LMIs in Theorem \ref{distributed_controller}. Specifically, the obtained controllers for System 1 are
\begin{equation*}
    \begin{aligned}
        &G_{1,1}^{1,1} =[-8.638\	-0.498],\ G_{1,1}^{1,2} =[-8.500\ -0.391],\\
        &G_{1,1}^{1,3} =[-8.610\	-0.477],\ G_{1,1}^{2,1} =[-4.878\ -0.501],\\
        &G_{1,1}^{2,2} =[-4.706\ -0.347],\ G_{1,1}^{2,3} =[-4.878\ -0.501],\\ &G_{1,2}^{1,1} =[-16.154\ 	-0.490],\ G_{1,2}^{1,2} =[-16.087\ -0.480],\\
        &G_{1,2}^{1,3} =[-16.076\ -0.431],\ G_{1,2}^{2,1} =[-19.913\ -0.487],\\
        &G_{1,2}^{2,2} =[-19.881\ -0.525],\ G_{1,2}^{2,3} =[-19.808\ -0.408],
    \end{aligned}
\end{equation*}
and the ones for System 2 are
\begin{equation*}
    \begin{aligned}
        &G_{2,1}^{1,1}=G_{2,1}^{1,2}=G_{2,1}^{1,3}=[-13.100\ -2.454\ 1.550],\\
        &G_{2,1}^{2,1}=G_{2,1}^{2,2}=G_{2,1}^{2,3}=[-17.592\ -0.798\ 5.666],\\
        &G_{2,2}^{1,1}=G_{2,2}^{1,2}=G_{2,2}^{1,3}=[-3.974\ -6.840\ 6.134],\\
        &G_{2,2}^{2,1}=G_{2,2}^{2,2}=G_{2,2}^{2,3}=[-4.071\ -7.606\ -5.580],\\
        &G_{2,3}^{1,1}=G_{2,3}^{1,2}=G_{2,3}^{1,3}=[0.427\ -23.902\ -22.903],\\
        &G_{2,2}^{2,1}=G_{2,2}^{2,2}=G_{2,2}^{2,3}=[0.266\ -23.881\ -22.386].
    \end{aligned}
\end{equation*}
 

\begin{figure}[t]
  \centering
  \subfigure[System 1's States]{
    \includegraphics[width=0.47\columnwidth]{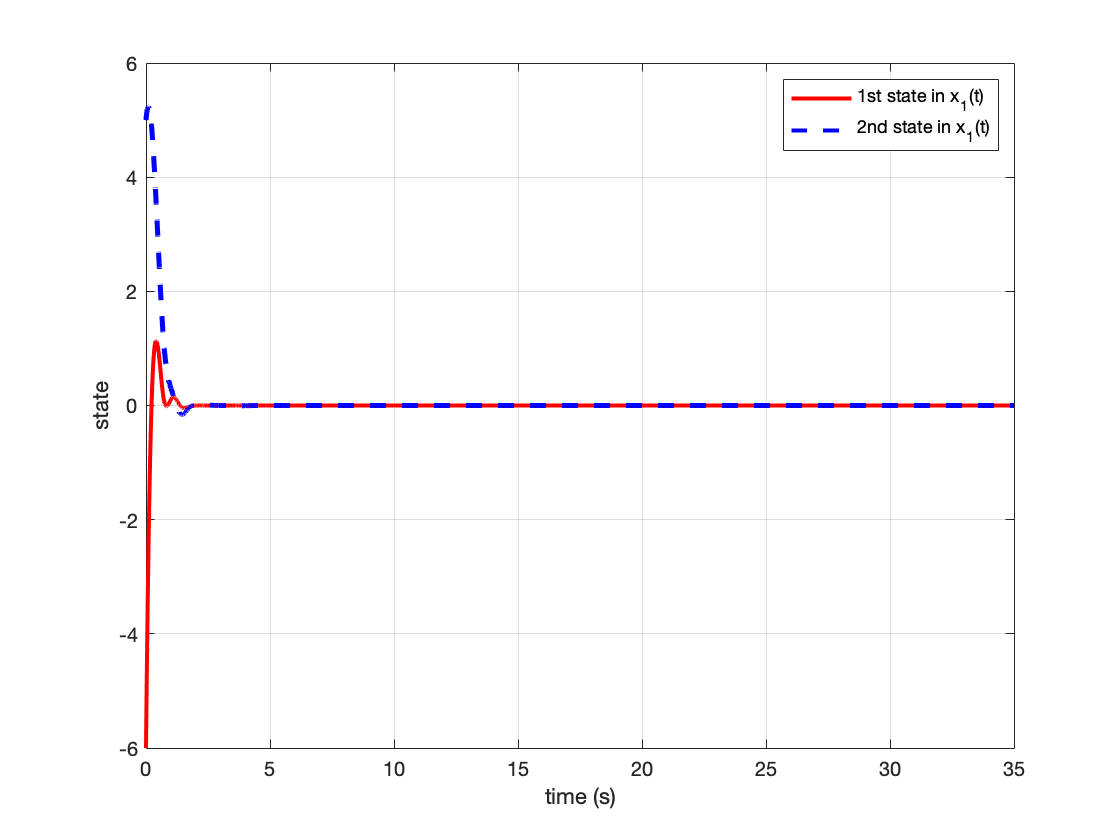}\label{state_1}}
	 \subfigure[System 2's States]{
    \includegraphics[width=0.47\columnwidth]{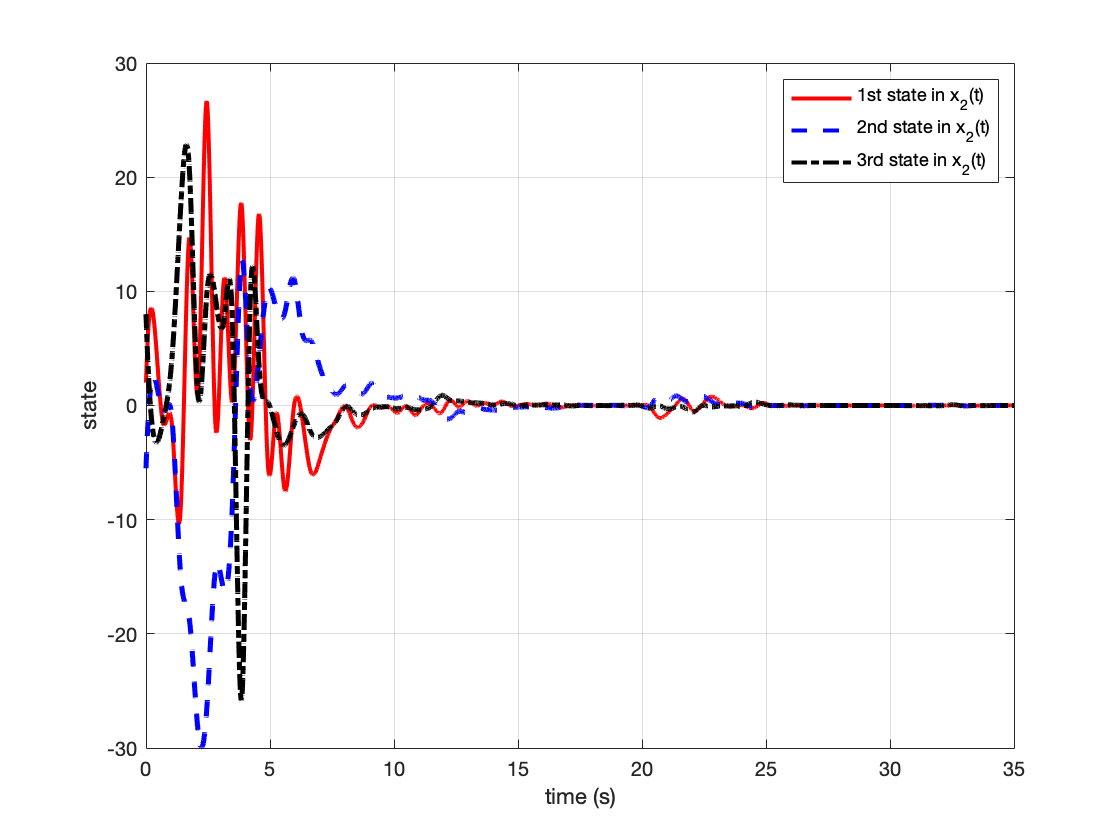}\label{state_2}}
  \caption[]{(a) and (b) show the stabilized state trajectories of System 1 and System 2, respectively, with the control designed under partial observation.}
  \label{sys_state}
\end{figure}


\begin{figure}[t]
  \centering
  \subfigure[System 1's Mode]{
    \includegraphics[width=0.47 \columnwidth]{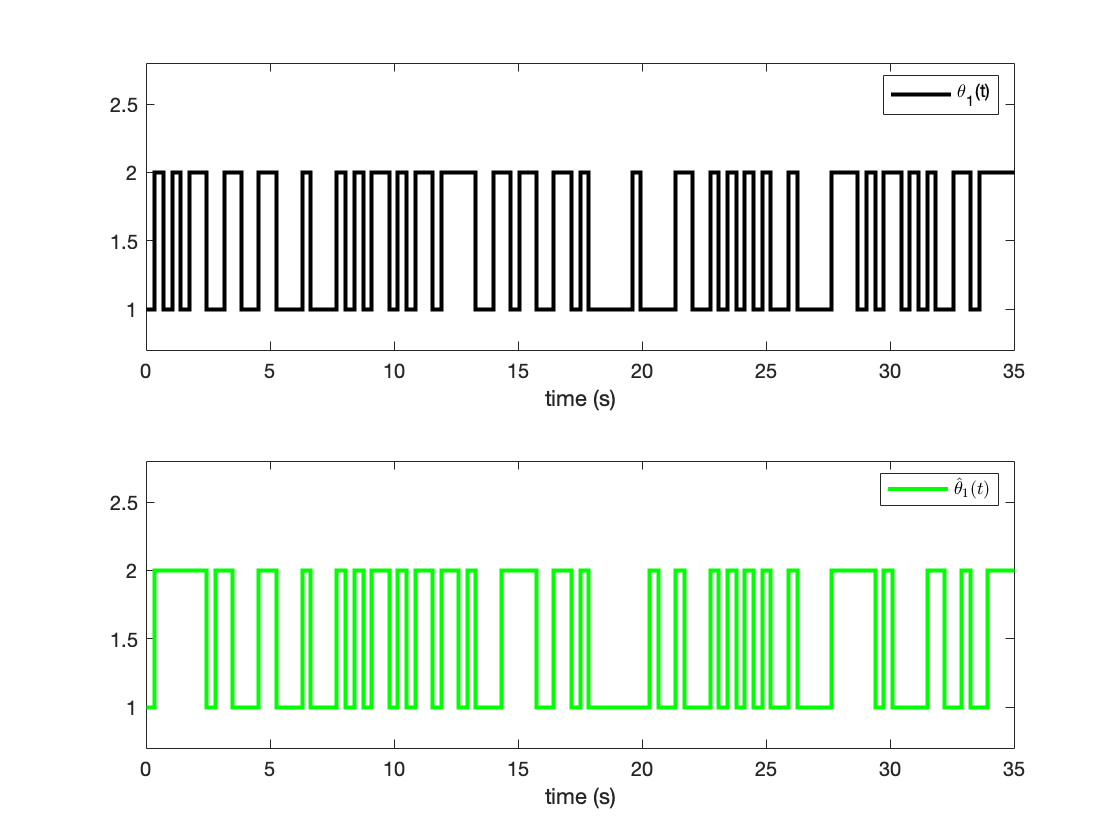}\label{markov_1}}
	 \subfigure[System 2's Mode]{
    \includegraphics[width=0.47\columnwidth]{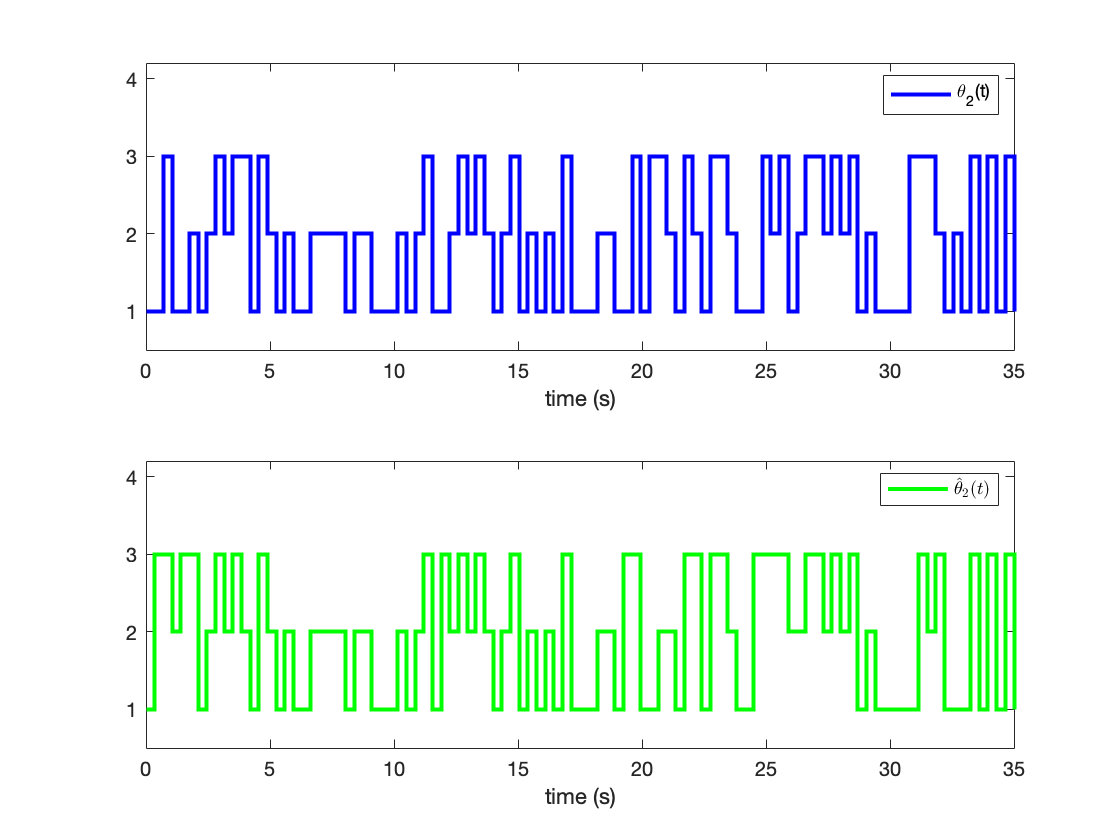}\label{markov_2}}
  \caption{The sampled Markov chains of System 1 and System 2, respectively.}
  \label{markov_state}
\end{figure}

\begin{figure}[t]
  \centering
  \subfigure[System 1's States]{
    \includegraphics[width=0.47\columnwidth]{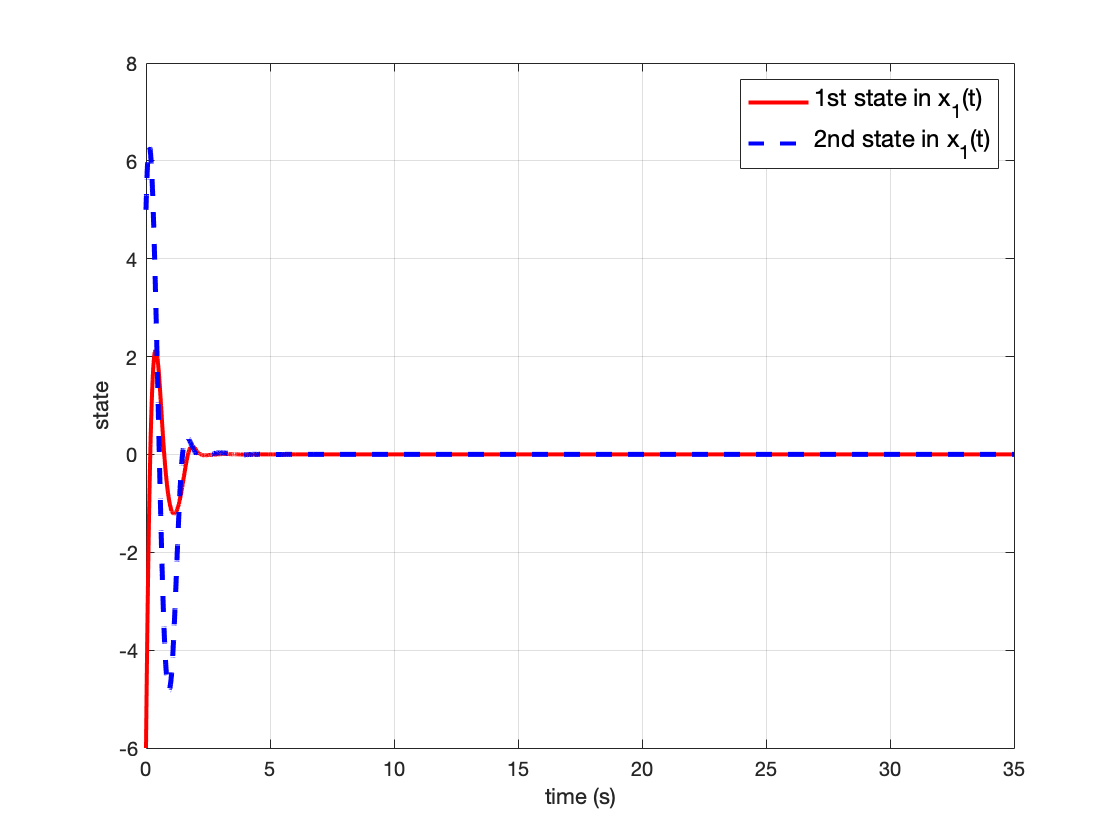}}
	 \subfigure[System 2's States]{
    \includegraphics[width=0.47\columnwidth]{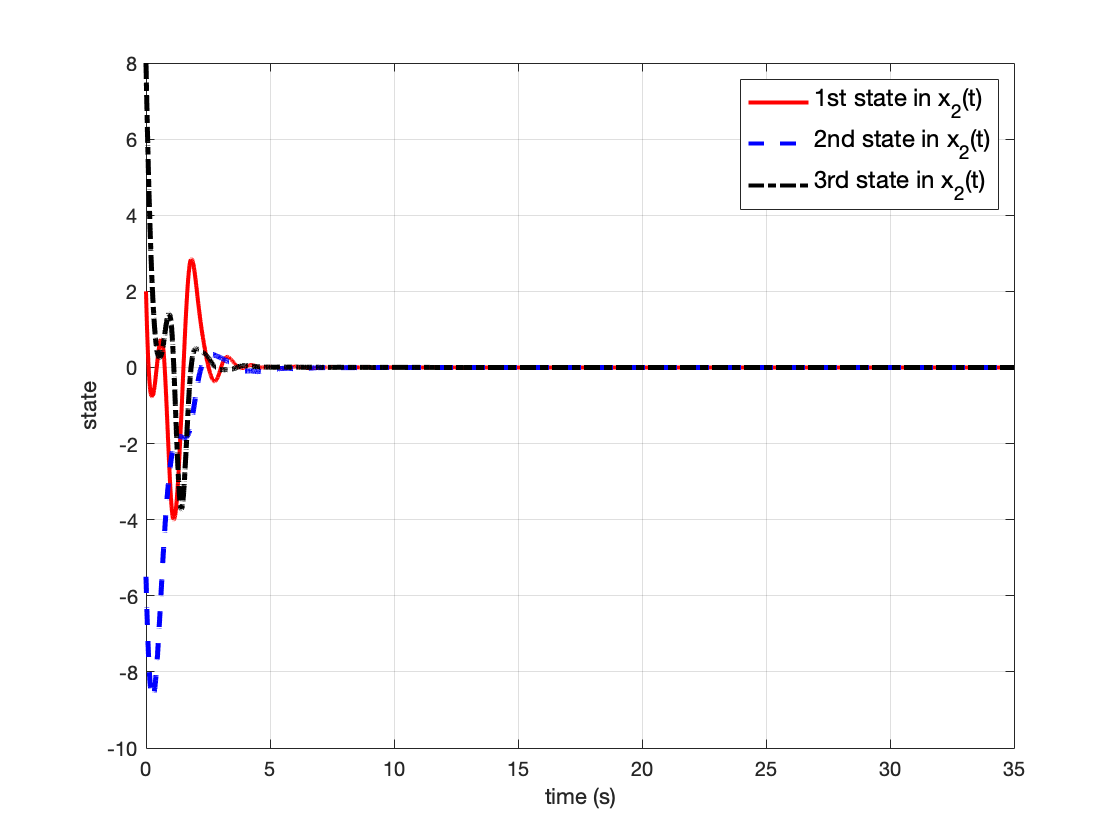}}
  \caption[]{{(a) and (b) show the stabilized state trajectories of System 1 and System 2, respectively, with the control designed under full observation.}}
  \label{state_full_info}
\end{figure}

With the designed controllers above, Fig. \ref{sys_state} shows the state trajectories of the interdependent systems with the initial conditions $x_1(0)=[-6, 5]^{\rm T}$, and $x_2(0)=[2, -5.5, 8]^{\rm T}$. Fig. \ref{markov_state} depicts the sampled Markov chains of the underlying parameters $\theta_1(t)$ and $\theta_2(t)$, and their observations $\hat{\theta}_1(t)$ and $\hat{\theta}_2(t)$, respectively. {For comparison, Fig. \ref{state_full_info} illustrates the results with the control designed under complete observation. With a perfect knowledge on the system mode, the state trajectories in Fig. \ref{state_full_info} are relatively smoother and converge to the origin faster than those in Fig. \ref{sys_state}. However, the advantage of the designed distributed control strategy lies in the fact that, though the systems' modes are not directly observable, it can still stabilize the interdependent MJLSs with satisfactory performance as shown in Fig. \ref{sys_state}.}

\section{Conclusion}\label{conclusion}
In this paper, we have studied the interdependent multiple MJLSs. We have designed distributed stabilizing controllers for each MJLS with partial information which only require the system state information and indirect observations of the local mode. In addition, these designed controllers are able to stabilize the integrated Markov jump system. The distributed feature of these controllers reduce the information exchange and communication costs between different Markov jump systems. {The future work would be extending the stabilizing control design to optimal control design considering state and control costs for the coupled MJLSs under incomplete information. }
 
\appendices
\section{Proof of Lemma \ref{simplify}}\label{appndix1}
Note that 
${x}(t+\Delta) = (I+\Delta\cdot {A}^m_{\theta(t)\hat{\theta}(t)}){x}(t) +  \Delta\cdot {D}_{\theta(t)}{w}(t),$ 
where $I$ is an identity matrix of appropriate dimension.
 By definition, the infinitesimal generator of $V$ is equal to
\begin{equation*}
    \begin{aligned}
&\mathcal{L}V({x}(t),{\theta}(t)) \\
&= \lim_{\Delta\rightarrow 0}  \frac{1}{\Delta} \Big\{ \mathbb{E}\big[V({x}(t+\Delta),{\theta}(t+\Delta))|{x}(t),\theta(t)\big] 
\\
& \qquad\qquad\qquad\qquad\qquad\qquad\qquad\qquad \quad - V({x}(t),{\theta}(t)) \Big\}\\
& = {x}^{\rm T}(t)\sum_{\hat{i}\in{\mathcal{S}}}\alpha_{i\hat{i}} \left( P_i{A}_{i\hat{i}}^{m} + {A}_{i\hat{i}}^{m{\rm T}}P_i +  \sum_{j\in\mathcal{S}} \gamma^{m}_{ij} P_j\right){x}(t)\\
&\qquad\qquad\qquad\qquad\qquad\qquad\qquad\qquad\ + 2{x}^{\rm T}(t) P_i {D}_i{w}(t)\\
& = {x}^{\rm T}(t)\left( P_i\bar{A}_{i}^{m} + \bar{A}_i^{m{\rm T}}P_i +  \sum_{j\in\mathcal{S}} \gamma^{m}_{ij} P_j\right){x}(t)\\
&\qquad\qquad\qquad\qquad\qquad\qquad\qquad\qquad\ + 2{x}^{\rm T}(t) P_i {D}_i{w}(t),
    \end{aligned}
\end{equation*}
{where the second equality is due to tower property and Assumption \ref{assump1}}.
\qed

\end{document}